\documentclass[oneside]{article}
\usepackage{setspace}
\usepackage{caption}
\usepackage{qic,epsfig}
\usepackage{amsfonts}
\usepackage[cmex10]{amsmath}
\usepackage{tabularx}
\usepackage{url}
\usepackage{bm}
\usepackage{supertabular}
\usepackage{longtable}
\usepackage{xcolor}
\usepackage{lineno}
\begin{document}
\setlength{\textheight}{8.0truein}    

\runninghead{On quantum tensor product codes}
            {J-H Fan, Y-H Li, M-H Hsieh,	and H-W Chen}

\normalsize\textlineskip
\thispagestyle{empty}
\setcounter{page}{1105}

\copyrightheading{17}{13\&14}{2017}{1105--1122}

\vspace*{0.88truein}

\alphfootnote

\fpage{1105}

\centerline{\bf
ON QUANTUM TENSOR PRODUCT CODES}
\vspace*{0.37truein}
\centerline{\footnotesize
JIHAO FAN}
\vspace*{0.015truein}
\centerline{\footnotesize\it Department of Computer Engineering, Nanjing Institute of Technology, Nanjing, Jiangsu 211167, China, $\&$
}
\centerline{\footnotesize\it School of Computer Science and Engineering, Southeast University, Nanjing, Jiangsu 211189, China
}
\centerline{\footnotesize\it \url{jihao.fan@outlook.com}}
\vspace*{10pt}
\centerline{\footnotesize
YONGHUI LI}
\vspace*{0.015truein}
\centerline{\footnotesize\it School of Electrical and Information Engineering,
the University of Sydney, Sydney, NSW 2006, Australia
}
\centerline{\footnotesize\it \url{yonghui.li@sydney.edu.au}}
\vspace*{10pt}
\centerline{\footnotesize
MIN-HSIU HSIEH}
\vspace*{0.015truein}
\centerline{\footnotesize\it Center for Quantum Software and
Information, Faculty of Engineering and Information Technology,
}
\centerline{\footnotesize\it  University of Technology Sydney, NSW 2007, Australia}
\centerline{\footnotesize\it \url{min-hsiu.hsieh@uts.edu.au}}
\vspace*{10pt}
\centerline{\footnotesize
HANWU CHEN}
\vspace*{0.015truein}
\centerline{\footnotesize\it School of Computer Science and Engineering, Southeast University, Nanjing, Jiangsu 211189, China}
\centerline{\footnotesize\it \url{hw_chen@seu.edu.cn}}
\vspace*{0.225truein}
\publisher{August 2, 2016}{August 21, 2017}

\vspace*{0.21truein}

\abstracts{
We present a general  framework for the construction of quantum tensor product codes (QTPC).
 In a classical tensor product code (TPC),  its parity check matrix is constructed via   the tensor  product of   parity check
matrices of the two component codes.
We show that by adding  some constraints  on  the component codes,    several classes  of dual-containing TPCs can be obtained.
By selecting different types of component codes, the proposed method enables the construction of a large family of QTPCs and they can provide a
wide variety of quantum error control abilities. In particular, if one of the component codes is selected as a burst-error-correction code, then   QTPCs  have quantum multiple-burst-error-correction abilities, provided these bursts fall in distinct subblocks. Compared with  concatenated quantum codes (CQC),   the component code selections of QTPCs  are much more  flexible than those of CQCs  since only one of the component codes of QTPCs needs to satisfy the dual-containing restriction.  We show that it is possible to construct QTPCs with  parameters better than other classes of quantum error-correction codes (QECC), e.g., CQCs and quantum BCH codes.  Many QTPCs  are obtained with parameters better than previously known quantum codes available   in the literature. Several classes of QTPCs that can correct multiple quantum bursts of errors are constructed based on reversible cyclic codes and maximum-distance-separable (MDS) codes.
}{}{}

\vspace*{10pt}

\keywords{Quantum error-correction codes, Tensor product codes, Burst-error-correction codes, Concatenated quantum codes (CQC),
Maximal-distance-separable (MDS) codes, BCH codes}
\vspace*{3pt}
\communicate{R~Jozsa~\&~R~Laflamme}

\vspace*{1pt}\textlineskip    

\section{Introduction}
\newtheorem{examples}{Example}
Quantum information is sensitive and vulnerable to quantum noise  during the process of quantum computations and
quantum communications. By employing redundancy, quantum error-correction codes (QECC)   can  provide the effective protection of  quantum information against  errors caused by  decoherence and other quantum noise. As shown in  the pioneering work in \cite{steane1996error,shor1995scheme}, it is possible to construct QECCs from classical  error-correction codes (ECC) that
subject to certain constraints. Furthermore, stabilizer codes  \cite{gottesman1997stabilizer,ketkar2006nonbinary,calderbank1998quantum,brun2006correcting,hsieh2007general,hsieh2011high} provide  a more general framework to construct QECCs analogous to classical additive codes.   

Classical tensor product  codes were first proposed by Wolf in the 1960's  \cite{wolf1965codes,wolf2006introduction} and
were later generalized in \cite{imai1981generalized}. The parity check matrix of a TPC is obtained by taking the tensor product of   parity check
matrices of the two component codes. Based on the choice of the   component codes,     TPCs can be designed to provide error-correction, error-detection or error-location
properties. Recently, several classes of TPCs  have been considered to be used in data storage systems, e.g., in magnetic recording   \cite{alhussien2010iteratively, chaichanavong2006tensor, chaichanavong2006tensor1},  in Flash memory     \cite{gabrys2013graded,kaynak2014classification} and in the construction of   locally
repairable codes  which are applied in  distributed storage systems \cite{huang2015linear,huang2015binary}.
 In \cite{alhussien2010iteratively},  an iteratively decodable  TPC by concatenating  an error-pattern correction code with a $q$-ary LDPC code was proposed.
  The tensor product  concatenating  scheme could  significantly  improve the efficiency  of the inner parity code while retaining a similar performance.
In  \cite{maucher2000equivalence,bossert1999some,Helmut2003,fan2017comments}, generalized TPCs (also called  generalized error location codes) were shown to be  equivalent to generalized concatenated codes.

In \cite{grassl2005quantum},
quantum block and convolutional codes based on self-orthogonal TPCs with component codes over the same field were constructed.  In \cite{la2012asymmetric}, asymmetric quantum product codes were constructed from the tensor product of two Reed-Solomon (RS) codes. In \cite{fan2017construction}, quantum error-locating codes which can indicate the location of quantum
errors in a single sub-block were constructed. We should emphasize that although quantum tensor product codes (QTPC) were first proposed by Grassl and R{\"o}tteler  in \cite{grassl2005quantum}, they had  only considered  the construction of QTPCs from classical TPCs with component codes over the same field.  The design of QTPCs from TPCs with one of the component codes over the extension field has not been considered.  However, TPCs with one of the component codes over the extension field are more important and have more practical applications, see \cite{wolf1965codes,wolf2006introduction,alhussien2010iteratively,chaichanavong2006tensor, chaichanavong2006tensor1,gabrys2013graded,kaynak2014classification,huang2015linear,huang2015binary}.

 In this paper,  we propose a generalized construction of a series of QTPCs based on classical TPCs with one of the component codes over the extension field. The proposed QTPCs   will  exhibit   quantum  error-correction, quantum error-detection or quantum error-location
properties as their classical counterparts.
We show that as long as  one of the component codes of  classical TPCs satisfies certain dual-containing condition,
the resultant TPC will satisfy the  dual-containing condition for constructing QECCs. As a result, the choice of the other
 component code can be selected  without the dual-containing restriction.
 Several classes of QTPCs with different error control abilities are obtained.
 Compared with other quantum concatenation schemes, such as concatenated quantum codes (CQC) \cite{knill1996concatenated,grassl2009generalized},
 the component code  selections  of QTPCs  are   more flexible than those of CQCs. As we know, CQCs cannot be constructed from classical concatenated codes directly. A CQC is usually constructed from two component QECCs: an outer QECC and an inner QECC \cite{knill1996concatenated,grassl2009generalized}, e.g., a quantum RS code as the outer code and a binary QECC as the inner code.
Therefore, both the outer and inner component codes need to be QECCs. By contrast,
 only one of the component codes of QTPCs needs to satisfy the dual-containing restriction.
Furthermore,  if   letting  the minimum distance be a comparable length, QTPCs can have dimensions much larger than the dimensions of CQCs with the same length and minimum distance. Two families of QTPCs with  better parameters  than CQCs are constructed.
Moreover, several families of QTPCs are obtained with parameters better than quantum BCH codes in \cite{aly2007quantum,ling2010generalization}  or QECCs with minimum distance five and six in \cite{li2008binary}.
 It is known that classical TPCs cannot have parameters better than classical BCH codes in classical coding theory. However, QTPCs can  have parameters better than quantum BCH codes in the quantum case.
 More recently,  a comprehensive survey in \cite{caruso2014quantum} discussed the memory effects in quantum channels and it was shown that these effects can be accurately described by correlated error models, for which quantum-burst-error-correction codes (QBECC) (see \cite{vatan1999spatially,kawabata2000quantum}) should be designed to cope with these correlated errors. However, the construction of QBECCs with single or multiple-burst-error-correction abilities has received less attention. We show that
QTPCs have quantum multiple-burst-error-correction abilities as
their classical counterparts if one of the component codes is chosen as a burst-error-correction code,
provided these bursts fall in distinct
subblocks.

The paper is organized as follows.  Section \ref{Preliminaries} gives a brief review  of QECCs and classical TPCs.  Section \ref{quantum_TPC}  proposes   a general  framework for the construction of QTPCs by investigating the dual-containing properties of classical TPCs. Then in Section \ref{Code Comparisons}, several families of QTPCs are constructed and are compared with other classes of QECCs.   Section \ref{burst QTPCs} gives a special class of QTPCs with multiple-burst-error-correction abilities. The decoding of QTPCs is given in Section \ref{QTPCs Decoding}. Conclusions and discussions are shown in Section \ref{Conclusions}.

 \section{Preliminaries}
\label{Preliminaries}
\noindent
In this section we first review  some basic definitions and facts of QECCs, followed by the introduction of classical TPCs. We only consider the binary QECCs, i.e., the qubit  systems.

Let $q=2$ (or $q=4$). Let $GF(q)$ denote the finite field with $q$ elements.  The trace mapping $\textrm{Tr}:GF(4)\rightarrow GF(2)$ is given by $\textrm{Tr}(\alpha)=\alpha+\alpha^2$. Denote by $GF(q^m)$ a field extension of degree $m$ of the
field $GF(q)$. Let $C$ be a   linear code over $GF(q)$, the dual code of   $C$ is denoted by
 \begin{equation}
 C^\bot=\{c\in GF(q)^n|x\cdot c=0,\ \forall x\in C\}.
 \end{equation}
 For two vectors $u=(u_1,u_2,\ldots,u_n)$, $v=(v_1,v_2,\ldots,v_n)\in GF(4)^n$,  the Hermitian inner product of vectors $u$ and $v$ is denoted by
 \begin{equation}
 (u,v)_{h}= u\cdot  \overline{v} =\sum_{i=1}^nu_i v_i^2,
 \end{equation}
 and the trace-Hermitian inner product of  $u$ and $v$ is denoted by
 \begin{equation}
 (u,v)_{th}=\textrm{Tr}(u\cdot  \overline{v})=\sum_{i=1}^n(u_i v_i^2+u_i^2v_i),
 \end{equation}
  where $\overline{v}=(v_1^2,v_2^2,\ldots,v_n^2)$ denotes the conjugate of   vector $v=(v_1,v_2,\ldots,v_n)$. Let $D$  be a classical additive code  over $GF(4)$,   the  Hermitian dual code of   $D$ is denoted by
  \begin{equation}
 D^{\bot_{h}}=\{v\in GF(4)^n|(u,v)_{h}=0,\forall u\in D\},
 \end{equation}
 and the trace-Hermitian dual code of $D$ is denoted by
  \begin{equation}
 D^{\bot_{th}}=\{y\in GF(4)^n|(x,y)_{th}=0,\forall x\in D\}.
 \end{equation}

\subsection{\textbf{Quantum Error-Correction Codes}}
\noindent
Let $\mathbb{C}$ denote the complex number field. For a positive integer $n$,
let $V_n=(\mathbb{C}^2)^{\otimes n}=\mathbb{C}^{2^n}$ be the $n$th tensor
product of $\mathbb{C}^2$ representing the quantum Hilbert space over $n$ qubits. We denote by $\{|x\rangle|x\in GF(2)\}$ the vectors of  an orthonormal basis of $\mathbb{C}^2$.
Let $a,b  \in GF(2) $. The unitary operators $X(a)$ and $Z(b)$ are defined by
$X(a)|x\rangle=|x+a\rangle$ and $Z(b)|x\rangle=(-1)^{b\cdot x}|x\rangle$, respectively. Let $\mathbf{a}=\{a_1,\ldots,a_n\}$ and $\mathbf{b}=\{b_1,\ldots,b_n\}$ be two vectors over $ GF(2)$. Denote by $X(\mathbf{a})= X(a_1)\otimes \cdots \otimes X(a_n)$ and $Z(\mathbf{b})= Z(b_1)\otimes \cdots \otimes Z(b_n)$ the tensor products of $n$ error operators. Then the set $E_n=\{X(\mathbf{a})Z(\mathbf{b})|\mathbf{a},\mathbf{b}\in GF(2)^n\}$ is an
  error basis on  the quantum Hilbert space $V_n$. The finite group $G_n=\{\pm X(\mathbf{a})Z(\mathbf{b})|\mathbf{a},\mathbf{b}\in GF(2)^n \}$ is the error group  associated with the   error basis $E_n$.

\begin{definition}
\label{definition of QECC}
A stabilizer   code $Q=((n,2^k))$   is a $2^k$-dimensional  subspace   of $V_n=\mathbb{C}^{2^n}$ that satisfies $Q=\bigcap\limits_{e\in S}\{v\in\mathbb{C}^{2^n}|ev=v\}$, for some commutative subgroup $S$ of $G_n$. If $Q$ has minimum distance $d$, then it is denoted by $Q=((n,2^k,d))$ $($or $Q=[[n,k,d]])$.
\end{definition}

According to \cite{calderbank1998quantum,nebe2006self}, each binary stabilizer code $Q$ (also called additive quantum code)  corresponds to a classical additive code $D$ that is self-orthogonal with respect to the trace-Hermitian inner product over $GF(4)$.

\begin{theorem}[{\cite{calderbank1998quantum,nebe2006self}}]
 \label{additive_quantum_codes}
 An $((n,K,d))$ additive quantum code exists if and only if there exists an additive code $D$ over $GF(4)$ of cardinality $|D|=2^n/K$ such that $D\subseteq D^{\bot_{th}}$, and
 $d=wt(D^{\bot_{th}}\backslash D)$ if $K>1$ $($and $d=wt(D^{\bot_{th}})$ if $K=1)$, where $wt(A)=\min\{wt(a)|a\in A\}$,  $\forall A\subseteq GF(4)^n$.
\end{theorem}

If the   additive code $D$ in Theorem \ref{additive_quantum_codes} happens to be linear over $GF(4)$, then the trace-Hermitian dual code of $D$ is equal to the Hermitian dual code of $D$, i.e., $D^{\bot_{th}}=D^{\bot_{h}}$. Thus, quantum codes can be constructed  from classical linear codes.

\begin{lemma}[{\cite{calderbank1996good, calderbank1998quantum}}]
\label{QECCs Constructions}
\hspace{0.1mm}
\begin{itemize}
\item[1).]\emph{(CSS Construction):} Let $C_1$ and $C_2$ be two binary linear
codes with parameters
$[n,k_1,d_1]$ and $[n,k_2,d_2]$ such that $C_2^{\bot}\subseteq C_1$. Then there exists an $[[n,k_1+k_2-n,  d]]$ QECC with minimum distance
$d=\min \{\text {wt}(c)|c \in (C_1\backslash C_2^{\bot}) \cup(C_2\backslash C_1^{\bot}) \}$ which is  pure to $\min\{d_1,d_2\}$. If $d_1$ and $d_2$ are less than the minimum distances of $C_2^{\bot}$ and $C_1^{\bot}$, respectively, then the QECC is pure and has minimum distance $\min\{d_1,d_2\}$.
\item[2).]\emph{(Hermitian Construction):} If there exists a quaternary $[n,k,d] $ linear code $D$ such that
$D^{\bot_h}\subseteq D$, then there exists an $[[n,2k-n,  d]]$ QECC with minimum distance $d=\min \{\text {wt}(c)|c \in D\backslash D^{\bot_h}\}$ which is pure to $d$. If $d$ is less than the minimum distance of $D^{\bot_h}$, then the QECC is pure and has minimum
distance $d$.
\end{itemize}
\end{lemma}

Let $E$ be a set of possible quantum errors that belong to $G_n$. Let $ \mathbf{e}=\{(\mathbf{a}|\mathbf{b} )|X(\mathbf{a})Z(\mathbf{b})\in E\ \text{or }-X(\mathbf{a})Z(\mathbf{b})\in E\}$ be the corresponding set of classical errors. Denote by
\begin{eqnarray}
\label{E-correction}
\nonumber
\mathbf{e}_X&=&\{\mathbf{a}\in GF(2)^n|\exists \mathbf{b}\in GF(2)^n, (\mathbf{a}|\mathbf{b})\in \mathbf{e}\},\\
\mathbf{e}_Z&=&\{\mathbf{b}\in GF(2)^n|\exists \mathbf{a}\in GF(2)^n, (\mathbf{a}|\mathbf{b})\in \mathbf{e}\},
\end{eqnarray}
where $\mathbf{e}_X$ and  $\mathbf{e}_Z$ are called bit error class and phase error class, respectively.  Then, the CSS  construction  is generalized to correct   any quantum errors of set $E$ in \cite{vatan1999spatially}.

\begin{lemma}[{\cite[Theorem 2]{vatan1999spatially}}]
\label{vatan_Constructions}
Let $E$ be a set of possible quantum errors. If there are $[n,k]$ classical codes $\mathcal{C}_1$ and  $\mathcal{C}_2$
such that  $\mathcal{C}_2^\bot\subseteq \mathcal{C}_1$ and $\mathcal{C}_1$ has $\mathbf{e}_X$-correction  ability,
$\mathcal{C}_2$ has $\mathbf{e}_Z$-correction  ability, then there exists an $[[n,2^{2k-n}]]$ QECC that has $E$-correction  ability.
\end{lemma}

Let $C_1$ and $C_2$ be two linear codes over $GF(q)$  with
parity check matrices   $H_1$ and $H_2$, respectively. It is easy to see that   $C_2^{\bot}\subseteq C_1$ is equivalent  to $H_1H_2^T=0$.
Let $D$ be a quaternary linear code with the parity check matrix $H$, then $D^{\bot_h}\subseteq D$ is equivalent  to $HH^\dagger=0$, where
the dagger ($\dagger$) denotes the conjugate transpose operation over matrices in $GF(4)$.

\subsection{\textbf{Classical Tensor Product Codes}}
\noindent
Let $C = [n,k,d]_q$ be a classical linear code over $GF(q)$, where $n$ is the code
length, $k$ is the dimension, $d$ is the minimum distance, and let $\rho = n - k$
be the number of   check symbols. 
 Define the random error-pattern class  $\xi_1$ as the class of error-patterns of weight less than or
equal to $r$. Define the burst error-pattern class  $\xi_2$ as  the class of error-patterns in which
the errors span no more than $b $ symbols. Let $H_{c_1}$ be the parity check matrix of  a linear code $C_1=[n_1, k_1, d_1]_q$, and the number of check
symbols is $\rho_1=n_1-k_1$. We assume $C_1$ corrects any error-pattern that belongs to class $\xi_i(i=1\text{ or }2)$. Let  $C_2=[n_2, k_2, d_2]_{q^{\rho_1}}$  be a linear code over the  extension field $GF(q^{\rho_1})$,  and the number of
check symbols is $\rho_2=n_2-k_2$. Let $H_{c_2}$ be the parity check matrix of $C_2$, and assume  $C_2$ corrects any error-pattern belongs to class $\zeta_i(i=1\text{ or }2)$, where $\zeta_1$ denotes a random error-pattern class and  $\zeta_2$ denotes a burst error-pattern class. We denote by
\begin{equation}
\mathcal{C}\equiv C_2\otimes_HC_1
\end{equation}
the tensor product code of $C_1$ and $C_2$\footnote{The direct product code defined by the direct product of the generator matrices of $C_1$ and $C_2$ is usually denoted by $C_1\otimes C_2$ (see \cite{macwilliams1981theory}). In order to  distinguish from that, we add a subscript `$H$' under the tensor product.}. If  we
consider  $H_{c_1}$  as a $1\times n_1$ matrix with elements
from $GF(q^{\rho_1})$,  then the parity check matrix $H_{\mathcal{C}}$ of   $\mathcal{C}$ is  the tensor product of $H_{c_1}$ and $H_{c_2}$ 
\begin{equation}
\label{Def_TPC}
H_{\mathcal{C}}=H_{c_2}\otimes H_{c_1}.
\end{equation}
Convert
the elements of $H_{\mathcal{C}}$ into $q$-ary columns with $GF(q)$ elements, then we can obtain the parity check matrix of $\mathcal{C}$ with $GF(q)$ elements.

\begin{lemma}[{\cite[Theorem 1]{wolf1965codes}}]
\label{wolftheorem}
If the codewords of $\mathcal{C}=C_2\otimes_HC_1$  consist  of $n_2$
subblocks, each subblock containing $n_1$   codewords, then
the code $\mathcal{C}$ can correct all error-patterns where the
subblocks containing errors form a pattern belonging to
class  $\xi_i(i=1\text{ or }2)$ and the errors within each erroneous subblock fall
within the class $\zeta_i(i=1\text{ or }2)$.
\end{lemma}

The elements of  the  parity check  matrix $H_{c_2}$ can also be  represented   by the companion matrices (see \cite{maucher2000equivalence,imai1981generalized}).  Let $f(x)=f_0+f_1x+\cdots+f_{\rho_1-1}x^{\rho_1-1}+x^{\rho_1}$ be a primitive polynomial of degree $\rho_1$ over $GF(2^{\rho_1})$ and let $\alpha$ be a primitive element of $GF(2^{\rho_1})$. The companion matrix $M$ of $f(x)$ is defined to be the $\rho_1\times \rho_1$ matrix, see Ref. \cite[Ch.4]{macwilliams1981theory},
  \begin{equation}
\label{companion_matrix}
M =
\left[
\begin{array}{ccccc}
 0&1&0&\cdots&0\\
 0&0&1&\cdots&0\\
\vdots&\vdots&\vdots&\ddots&\vdots\\
 0&0&0&\cdots&1\\
f_0&f_1&f_2&\cdots&f_{\rho_1-1}
\end{array}
\right]_{\rho_1\times \rho_1}
\end{equation}
Then for any element $a=\alpha^i$ of $GF(2^{\rho_1})$, the companion matrix of $a $ is denoted by $[a]=M^i$, a $\rho_1\times \rho_1$ matrix with $GF(2)$ elements.
Denote   the parity check matrix of the component code $C_2$ by $H_{c_2}=(b_{ij})_{\rho_2\times n_2}$   with $GF(2^{\rho_1})$ elements, i.e., $b_{ij}\in GF(2^{\rho_1})$ for $1\leq i \leq \rho_2$ and $1\leq j\leq n_2$. We use the notation in \cite{maucher2000equivalence} and denote by $[H_{c_2}]=([b_{ij}])_{\rho_1\rho_2\times \rho_1n_2}$ the companion matrix representation of $H_{c_2}$. 
Then the parity check matrix $H_{\mathcal{C}}$ of a binary TPC can be defined by
\begin{eqnarray}
\label{companion_matrix_TPC}
\nonumber
 &&H_{\mathcal{[C]}} \equiv [H_{c_2}^t]\otimes  H_{c_1} \\
&&= \left[
\begin{array}{cccc}
[b_{11}^t]H_{c_1}& [b_{12}^t]H_{c_1}&\cdots& [b_{1n_2}^t]H_{c_1} \\

[b_{21}^t]H_{c_1}& [b_{22}^t]H_{c_1}&\cdots& [b_{2n_2}^t]H_{c_1} \\
\vdots&\vdots&\vdots&\vdots\\

[b_{\rho_21}^t]H_{c_1}& [b_{\rho_22}^t]H_{c_1}&\cdots& [b_{\rho_2n_2}^t]H_{c_1}
\end{array}
\right],
\end{eqnarray}
where the matrix $[H_{c_2}^t]$ is obtained by transposing the component companion matrices  of $[H_{c_2} ]$, and $[b_{ij}^t]$ denotes by the transpose of $[b_{ij}]$.

According to \cite{lidl1997finite}, the companion matrix can also be defined to be the transpose of (\ref{companion_matrix}), alternately. Therefore, the parity check matrix $H_{\mathcal{C}}$ of a binary TPC can also be defined alternately without transposing the companion matrices, i.e.,
\begin{eqnarray}
\label{companion_matrix_TPC2}
\nonumber
 &&H_{\mathcal{[C]}} \equiv  [H_{c_2}]\otimes  H_{c_1} \\
&&= \left[
\begin{array}{cccc}
[b_{11}]H_{c_1}& [b_{12}]H_{c_1}&\cdots& [b_{1n_2}]H_{c_1} \\

[b_{21}]H_{c_1}& [b_{22}]H_{c_1}&\cdots& [b_{2n_2}]H_{c_1} \\
\vdots&\vdots&\vdots&\vdots\\

[b_{\rho_21}]H_{c_1}& [b_{\rho_22}]H_{c_1}&\cdots& [b_{\rho_2n_2}]H_{c_1}
\end{array}
\right].
\end{eqnarray}

Typically, the minimum distance of the resultant TPCs is bounded by the minimum distance  of the component codes, and the code length  and the number of check symbols get multiplied.

\begin{lemma}[{\cite{wolf1965codes,imai1981generalized}}]
\label{MDOfTPC}
Let $C_1=[n_1,k_1,d_1]_q$ and $C_2=[n_2,k_2,d_2]_{q^{\rho_1}}$ be two linear codes, and the numbers of check symbols are  $\rho_1=n_1-k_1$ and $\rho_2=n_2-k_2$, respectively. Then the tensor product code $\mathcal{C}=C_2\otimes_HC_1$ has parameters $[n_1n_2,n_1n_2-\rho_1\rho_2,\min\{d_1,d_2\}]_q$.
\end{lemma}

\section{Quantum Tensor Product Codes}
\label{quantum_TPC}
\noindent
In general,  a classical TPC can be constructed from  arbitrary two shorter codes over the same field or with one of the component codes over the extension field. In \cite{grassl2005quantum}, QTPCs based on self-orthogonal tensor product codes with component codes over the same field were constructed.

\begin{lemma}[{\cite{grassl2005quantum}}]
\label{TPC_Same_Fields}
Let $H_{c_1}$  and $H_{c_2}$ be the parity check matrices  of two classical linear codes $C_1$ and $C_2$ over $GF(q)$,  respectively.  Denote by $H=H_1\otimes H_2$ the tensor product of $H_1$ and $H_2$. Let  $\mathcal{C} $ be the tensor product code of $C_1$ and $C_2$ with the parity check matrix   given by  $H$, then $\mathcal{C}^\bot\subseteq\mathcal{C}$   if and only if $C_1^{\bot} \subseteq C_1$ $($or $C_2^\bot\subseteq C_2)$,  and $\mathcal{C}^{\bot_h}\subseteq\mathcal{C}$    if and only if   $C_1^{\bot_h} \subseteq C_1$ $($or  $C_2^{\bot_h}\subseteq C_2)$.
\end{lemma}

While considering TPCs with one of the component codes over the extension field, we show that if one of the component codes satisfies   certain dual-containing conditions, the resultant TPCs can also be dual contained.

Let $\alpha$ be a primitive element of $GF(q^{\rho_1})$  and assume that the basis used for   vector representation of elements in $GF(q^{\rho_1})$ is $\{1,\alpha,\alpha^2,\ldots,\alpha^{\rho_1-1}\}$.  We define the bijective  map \cite{macwilliams1981theory,lidl1997finite} $\psi: GF(q^{\rho_1})  \mapsto  GF(q)^{\rho_1} $ as
\begin{eqnarray}
  \nonumber
 \psi(\beta)&=&\psi(a_0+a_1\alpha+\cdots+a_{\rho_1-1}\alpha^{\rho_1-1})\\
 &\equiv&(a_0,a_1,\ldots,a_{\rho_1-1})^T,
\end{eqnarray}
where $\beta= a_0+a_1\alpha+\cdots+a_{\rho_1-1}\alpha^{\rho_1-1}$ is an element of $GF(q^{\rho_1})$, $a_i\in GF(q)$  for $0\leq i\leq \rho_1-1$, and $T$ is a  transpose of vector.

For the parity check matrix $H=(\alpha_{ij})_{k\times n}$  of code $C$ over $GF(q^{\rho_1})$,  we denote by
\[\psi(H)\equiv\big(\psi(\alpha_{ij}) \big)_{k\rho_1\times n}\] the matrix with elements converted from $H$ under $\psi$, and we denote by  $ \psi(C)$  the  corresponding subfield subcode with parity check matrix $ \psi(H)$. The inverse map of $\psi$ can be defined  as \begin{equation}\psi^{-1}:\
  GF(q)^{\rho_1}   \mapsto GF(q^{\rho_1})\end{equation} based on the same basis $\{1,\alpha,\alpha^2,\ldots,\alpha^{\rho_1-1}\}$. For the matrix $H_{c_1}=(c_{ij})_{\rho_1\times n}$  with $GF(q)$ elements, we denote by \[ \psi^{-1}(H_{c_1})\equiv\big(\psi^{-1}(col_{j}) \big)_{1 \times n}\] the matrix with elements converted from $H_{c_1}$ under $\psi^{-1}$, where $col_{j}=(c_{1j}, c_{2j},\ldots,c_{\rho_1j})^T$  for $1\leq j\leq n$.
Then the $q$-ary $\rho_1\times n_1$ parity check matrix $H_{c_1}$ of $C_1$ can   be
considered as a  $1\times n_1$  matrix $\psi^{-1}(H_{c_1})$ with elements from $GF(q^{\rho_1})$,  and $H_{\mathcal{C}}=H_{c_2}\otimes \psi^{-1}(H_{c_1})$  is a $\rho_2\times n_1n_2$ matrix with elements from $GF(q^{\rho_1})$.  If we replace the elements of $H_{\mathcal{C}}$ by   $\rho_1$-tuples   based on the same basis over $GF(q^{\rho_1})$, then we can get  a  $\rho_1\rho_2\times n_1n_2$ matrix $\psi(H_{\mathcal{C}})$ with  $GF(q)$ elements. Then the null space of the $\rho_1\rho_2\times n_1n_2$ $q$-ary matrix
$\psi(H_{\mathcal{C}})$ corresponds to a  $q$-ary TPC with parameters $\mathcal{C}=[n_1n_2,n_1n_2-\rho_1\rho_2]_q$.

\begin{lemma}
\label{CSS_TPC1}
Let $C_1=[n_1,k_1,d_1]_q$ be a $q$-ary linear code, and let $C_2=[n_2,k_2,d_2]_{q^{\rho_1}}$ be a linear code  over the extension field $GF(q^{\rho_1})$, and the numbers of   check symbols are $\rho_1=n_1-k_1$ and $\rho_2=n_2-k_2$, respectively. Let $\mathcal{C}=C_2\otimes_H C_1$ be the tensor product code of $C_1$ and $C_2$. If $C_1^\perp\subseteq C_1$ $($or $\psi(C_2)^\perp\subseteq\psi(C_2))$, then  $\mathcal{C}^\perp\subseteq\mathcal{C}$;  If $C_1^{\perp_h}\subseteq C_1$ $($or $\psi(C_2)^{\perp_h}\subseteq\psi(C_2))$, then $\mathcal{C}^{\perp_h}\subseteq\mathcal{C}$. 
\end{lemma}
\begin{proof}
Let   $H_{c_1}=(a_{ij})_{\rho_1\times n_1}$ be the parity check matrix of $C_1$ with $GF(q)$ elements.  Then $\psi^{-1}(H_{c_1})= [ \alpha_1\ \alpha_2\ \ldots\ \alpha_{n_1}  ]$ is a $1\times n_1$ array over $GF(q^{\rho_1})$. Let $H_{c_2}=(\beta_{ij})_{\rho_2\times n_2}$ be the parity check matrix of $C_2$ with elements from the extension field $GF(q^{\rho_1})$. Then  the tensor product of $H_{c_2}$ and $\psi^{-1}(H_{c_1})$ is given by
\begin{eqnarray}
\label{Parity_Check_TPC}
H_{\mathcal{C}}&=&H_{c_2}\otimes\psi^{-1}(H_{c_1})\nonumber\\
&=&\big[\beta_{ij}[\begin{array}{cccc} \alpha_1&\alpha_2&\ldots&\alpha_{n_1} \end{array} ]\big]_{\rho_2\times n_1n_2},
\end{eqnarray}
and we have
\begin{eqnarray}
\label{Binary_Parity_Check_TPC}\psi(H_{\mathcal{C}})=\left[\psi\big(\beta_{ij}[\begin{array}{cccc} \alpha_1&\alpha_2&\ldots&\alpha_{n_1} \end{array} ]\big)\right]_{\rho_1\rho_2\times n_1n_2}.
\end{eqnarray}
\begin{itemize}
\item[i)]
Since $H_{c_1}$ is a parity check matrix of $C_1$ over $GF(q)$, then $\psi^{-1}(H_{c_1})$ is  a parity check matrix over $GF(q^{\rho_1})$. It is easy to see that $\beta_{ij}\psi^{-1}(H_{c_1})$ is also a parity check matrix of $C_1$ over $GF(q^{\rho_1})$, where $1\leq i\leq\rho_2$ and $1\leq j\leq n_2$. Therefore,  $\psi(\beta_{ij}\psi^{-1}(H_{c_1}))$ is  also a parity check matrix of $C_1$ over $GF(q)$. Then we have
$\psi(\beta_{ij}\psi^{-1}(H_{c_1}))=A_{ij}\cdot H_{c_1}$, where $A_{ij}$ is an invertible $\rho_1\times\rho_1$ matrix over $GF(q)$.
Therefore,  $\psi(\beta_{ij}\psi^{-1}(H_{c_1}))\psi(\beta_{st}\psi^{-1}(H_{c_1}))^T=A_{ij}H_{c_1}H_{c_1}^TA_{st}^T=0$ if and only if $H_{c_1}H_{c_1}^T=0$ for $1\leq i,s\leq \rho_2$ and $1\leq j,t\leq n_2$. It follows that $\psi(H_{\mathcal{C}})\psi(H_{\mathcal{C}})^T=0$.
\item[ii)]
We rearrange  columns of the parity check matrix (\ref{Parity_Check_TPC}) as follows
$H'_{\mathcal{C}} = \big[\alpha_{i}H_{c_2}\big]_{\rho_2\times n_1n_2} = \psi^{-1}(H_{c_1})\otimes H_{c_2}$.
It follows that $\psi(H'_{\mathcal{C}})=\left[\psi\big(\alpha_{i}H_{c_2}\big)\right]_{\rho_1\rho_2\times n_1n_2}$. Then we have  $\psi\big(\alpha_{i}H_{c_2}\big)=B_{i}\cdot \psi\big(H_{c_2}\big)$, where $B_{i}$ is an invertible $\rho_1\rho_2\times\rho_1\rho_2$ matrix over $GF(q)$. Therefore,  $\psi\big(\alpha_{i}H_{c_2}\big)\psi\big(\alpha_{i}H_{c_2}\big)^T=
B_{i}  \psi\big(H_{c_2}\big)\psi\big(H_{c_2}\big)^TB_{i}^T=0$ if and only if $\psi\big(H_{c_2}\big)\psi\big(H_{c_2}\big)^T=0$.
It follows that $\psi(H'_{\mathcal{C}})\psi(H'_{\mathcal{C}})^T=0$.  Since $H'_{\mathcal{C}} $ contains the same columns of $H_{\mathcal{C}}$ with different permutations of columns,   we also have $\psi(H _{\mathcal{C}})\psi(H _{\mathcal{C}})^T=0$.
\end{itemize}
Therefore, we have $\mathcal{C}^\perp\subseteq\mathcal{C}$.
If  $C_1^{\perp_h}\subseteq C_1$ (or $\psi(C_2)^{\perp_h}\subseteq\psi(C_2)$), then $\mathcal{C}^{\perp_h}\subseteq\mathcal{C}$ can be obtained similarly. \qed
\end{proof}

From  Lemma \ref{CSS_TPC1}, we know that if the component code $C_1$ is dual (or Hermitian dual) contained, then there are no further restrictions such as the dual-containing restrictions on the component code $C_2$. Therefore, we can combine an arbitrary   linear code over the extension field, such as an MDS code or a $q^{\rho_1}$-ary LDPC code, with  a dual-containing binary or quaternary linear code  to construct   QTPCs using the CSS construction or Hermitian construction.  On the other hand, if the component code $C_2$  satisfies $\psi(C_2)^\perp\subseteq\psi(C_2)$ (or $\psi(C_2)^{\perp_h}\subseteq\psi(C_2)$), then  $C_1$ can be chosen arbitrarily  without the dual-containing restriction.

\begin{theorem}
\label{CSS_TPC1_theorem}
Let $C_1=[n_1,k_1,d_1]_q$ be a $q$-ary linear code, and let $C_2=[n_2,k_2,d_2]_{q^{\rho_1}}$ be a linear code  over the extension field $GF(q^{\rho_1})$, and the numbers of   check symbols are $\rho_1=n_1-k_1$ and $\rho_2=n_2-k_2$, respectively. 
If $q=2$ and $C_1^\perp\subseteq C_1$  $($or $\psi(C_2)^\perp\subseteq\psi(C_2))$, then there exists a pure QTPC with parameters $\mathcal{Q}=[[n_1n_2, n_1n_2-2\rho_1\rho_2,\min\{d_1,d_2\}]]$. If $q=4$ and $C_1^{\perp_h}\subseteq C_1$  $($or $\psi(C_2)^{\perp_h}\subseteq\psi(C_2))$, then there also exists a pure QTPC with the same parameters.
\end{theorem}
\begin{proof}
Let $\mathcal{C}=C_2\otimes_H C_1$ be the tensor product code of $C_1$ and $C_2$. If $q=2$ and if $C_1^\perp\subseteq C_1$ or $\psi(C_2)^\perp\subseteq\psi(C_2)$, then $\mathcal{C}^{\perp}\subseteq\mathcal{C}$  by Lemma \ref{CSS_TPC1}. Combining Lemma
\ref{QECCs Constructions} and Lemma \ref{MDOfTPC}, we know that there exists a QTPC with parameters $\mathcal{Q}=[[n_1n_2, n_1n_2-2\rho_1\rho_2,\geq\min\{d_1,d_2\}]]$. Let $H_{c_1}$  and $H_{c_2}$ be the parity check matrices  of   $C_1$ and $C_2$,  respectively, then $\mathcal{C}^{\bot}$ has a generator matrix $[H_{c_2}]\otimes H_{c_1}$, and $\mathcal{C}^{\bot}$ is a concatenated code with $C_1^{\bot}$ as the inner code and with $C_2^{\bot}$ as the outer code from \cite{maucher2000equivalence}. Therefore, $\mathcal{C}^{\bot}$ has minimum distance $d^{\bot}=d_1^{\bot}d_2^{\bot}$, where $d_1^{\bot}$ and $d_2^{\bot}$ are the dual minimum distances of $C_1$ and $C_2$, respectively. Since $C_1^\perp\subseteq C_1$ or $\psi(C_2)^\perp\subseteq\psi(C_2)$, there must be $d^{\bot}>\min\{d_1,d_2\}$, if $d_1^{\bot}, d_2^{\bot} >1$. Hence $\mathcal{Q}$ is pure and has exact minimum distance equal to $\min\{d_1,d_2\}$.

If $q=4$ and if  $C_1^{\perp_h}\subseteq C_1$ or $\psi(C_2)^{\perp_h}\subseteq\psi(C_2)$, then $\mathcal{C}^{\perp_h}\subseteq\mathcal{C}$  by Lemma \ref{CSS_TPC1}. We can obtain a pure QTPC with parameters $\mathcal{Q}=[[n_1n_2, n_1n_2-2\rho_1\rho_2,\min\{d_1,d_2\}]]$ by using the Hermitian construction.  \qed
\end{proof}

\newcommand{\tabincell}[2]{\begin{tabular}{@{}#1@{}}#2\end{tabular}}
\begin{table}
\renewcommand{\arraystretch}{1.3}
\setlength{\tabcolsep}{5pt}
\tcaption{Comparisons between QTPCs  and CQCs, where QTPCs are constructed based on dual-containing BCH codes and MDS codes, and CQCs are obtained by Ref. \cite{knill1996concatenated,grassl2009generalized}.}
\label{Comparisons Between QTPCs and CQCs I}
\centering
\footnotesize
\begin{tabular}[c]{|c|c|c|c|l|l|l|l|}
\hline
$m$&$n_1$&$n_2$&$\delta_1$&$\eta_1$&$\eta_2$&\multicolumn{1}{c|}{QTPCs $\mathcal{Q}$}&\multicolumn{1}{c|}{CQCs $C_Q$ } \\
\hline
$5$&$31$&$23 \leq n_2\leq 2^{15}+1 $&$7$&$2$&$3$&$[[31n_2,31n_2-180,7]]$&$ [[31n_2,28n_2-112,6]]$ \\
$5$&$31$&$13 \leq n_2\leq 2^{15}+1 $&$6$&$2$&$3$&$ [[31n_2,31n_2-150,6]]$&$ [[31n_2,28n_2-112,6]]$ \\
$5$&$31$&$8  \leq n_2\leq 2^{10}+1 $&$5$&$2$&$2$&$ [[31n_2,31n_2-80,5]]$&$ [[31n_2,28n_2-56,4]]$ \\
$5$&$31$&$2\leq n_2\leq 2^{10}+1 $&$4$&$2$&$2$&$ [[31n_2,31n_2-60,4]]$&$ [[31n_2,28n_2-56,4]]$ \\
\hline
$6$&$63$&$7 \leq n_2\leq 4^{15}+1 $&$7$&$2$&$3$&$ [[63n_2,63n_2-180,7]]$&$ [[63n_2,60n_2-240,6]]$ \\
$6$&$63$&$6 \leq n_2\leq 4^{15}+1 $&$6$&$2$&$3$&$ [[63n_2,63n_2-150,6]]$&$ [[63n_2,60n_2-240,6]]$ \\
$6$&$63$&$5 \leq n_2\leq 4^{9}+1 $&$5$&$2$&$2$&$ [[63n_2,63n_2-72,5]]$&$ [[63n_2,60n_2-120,4]]$ \\
$6$&$63$&$4 \leq n_2\leq 4^{9}+1 $&$4$&$2$&$2$&$ [[63n_2,63n_2-54,4]]$&$ [[63n_2,60n_2-120,4]]$ \\
\hline
$7$&$127$&$34 \leq n_2\leq 2^{49}+1 $&$15$&$3$&$5$&$ [[127n_2,127n_2-1372,15]]$&$ [[127n_2,113n_2-904,15]]$ \\
$7$&$127$&$14 \leq n_2\leq 2^{49}+1 $&$14$&$2$&$7$ &$ [[127n_2,127n_2-1274,14]]$&$ [[127n_2,124n_2-1488,14]]$\\
$7$&$127$& $13 \leq n_2\leq 2^{42}+1$ &$13$&$2$&$6$ &$ [[127n_2,127n_2-1008,13]]$&$ [[127n_2,124n_2-1240,12]]$\\
$7$&$127$& $18 \leq n_2\leq 2^{42}+1$ &$12$&$3$&$4$ &$ [[127n_2,127n_2-924,12]]$&$ [[127n_2,113n_2-678,12]]$\\
$7$&$127$& $11 \leq n_2\leq 2^{35}+1$ &$11$&$2$&$5$ &$ [[127n_2,127n_2-700,11]]$&$ [[127n_2,124n_2-992,10]]$\\
$7$&$127$& $10 \leq n_2\leq 2^{35}+1$ &$10$&$2$&$5$ &$ [[127n_2,127n_2-630,10]]$&$ [[127n_2,124n_2-992,10]]$\\
\hline
\end{tabular}
\end{table}

If we use the companion matrix representation to represent the parity check matrix of a binary TPC, we can get the following result.

\begin{theorem}
\label{CSS_TPC2}
Let $H_{c_1}$ be the parity check matrix of a binary linear code $C_1=[n_1,k_1,d_1] $, and let $H_{c_2}$ be the parity check matrix of a linear code $C_2=[n_2,k_2,d_2]_{2^{\rho_1}}$ over the extension field $GF(2^{\rho_1})$, where $\rho_1=n_1-k_1$ is the number of check symbols of $C_1$.
 If the component code $C_2$ satisfies $C_2^\bot\subseteq C_2$, and $H_{c_1}H_{c_1}^T$ is of full rank, then
  there exists a pure QTPC with parameters $\mathcal{Q}=[[n_1n_2, n_1n_2-2\rho_1\rho_2,\min\{d_1,d_2\}]]$.
\end{theorem}
\begin{proof}
Let $\mathcal{C}=C_2\otimes C_1$ be the tensor product code of $C_1$ and $C_2$, and let $H_\mathcal{[C]}$ in (\ref{companion_matrix_TPC}) be a parity check matrix of $\mathcal{C}$.
Denote the product of $H_{c_1}$ and $H_{c_1}^T$ by $S=H_{c_1}H_{c_1}^T$. 
We know that there   exists an invertible matrix $L$   such that $LS $ is an identity matrix, i.e.,
$LS  =I_{\rho_1}$,
where $I_{\rho_1}$ is a $\rho_1\times \rho_1$ identity matrix. It follows that $LH_{c_1}H_{c_1}^T= LS =  I_{\rho_1}$.
Let $\mathcal{C}_L$ be a TPC which has the same parameters with $\mathcal{C}$ by replacing  the component matrix $H_{c_1}$ in  $H_\mathcal{[C]}$ with $LH_{c_1}$.  Denote by   $H_{[\mathcal{C}_L]}$   the parity check matrix of   $\mathcal{C}_L$ and we use (\ref{companion_matrix_TPC2}) with untransposed companion matrices as the parity check matrix  of $\mathcal{C}_L$.
Then  $\mathcal{C}$ and $\mathcal{C}_L$ have the same parameters and error control abilities but have    different parity  check matrix structures.
  If $C_2^{\bot}\subseteq C_2$, then $H_{c_2}H_{c_2}^T=0$.  It is easy to very   that this is equal to $[H_{c_2}][H_{c_2}^t]^T=0$.  It follows that $\sum_{k=1}^{n_2}[b_{ik}][b_{jk}]=0$ for all $1\leq i,j\leq \rho_2$. Then there is
\begin{eqnarray}
\nonumber
\sum_{k=1}^{n_2}([b_{ik}]LH_{c_1})([b_{jk}^t]H_{c_1})^T &=&\sum_{k=1}^{n_2}[b_{ik}]LH_{c_1}H_{c_1}^T[b_{jk}] \\
\nonumber
&=&\sum_{k=1}^{n_2}[b_{ik}] [b_{jk}] \\
&=&0
\end{eqnarray}
for all $1\leq i,j\leq \rho_2$. Therefore, we have
 $H_{[\mathcal{C}_L]}H_{[\mathcal{C}]}^T=0$  which follows that $\mathcal{C}_L^\perp\subseteq\mathcal{C} $.

By combining Lemma \ref{QECCs Constructions} and Lemma \ref{MDOfTPC}, we know that there exists a QTPC with parameters $\mathcal{Q}=[[n_1n_2, n_1n_2-2\rho_1\rho_2,\geq\min\{d_1,d_2\}]]$. Notice that $d_{\mathcal{C}_L^{\bot}}=d_{\mathcal{C}^{\perp}}
=d_{C_1^{\perp}}d_{C_2^{\perp}}>\min\{d_1,d_2\}$, thus $\mathcal{Q}$ is pure and has exact minimum distance $\min\{d_1,d_2\}$.
 \qed
\end{proof}

The  dual-containing restriction   on $C_2$ in Theorem  \ref{CSS_TPC1_theorem} is easy to be verified for RS codes, but for other nonbinary codes it becomes difficult to verify them.  In Theorem \ref{CSS_TPC2}, the restriction on $C_2$ is much easier to be satisfied than that in Theorem \ref{CSS_TPC1_theorem}, but one condition  is that the product matrix $H_{c_1}H_{c_1}^T$ needs to be a full rank.

\section{Code Constructions and Comparisons}
\label{Code Comparisons}
\noindent
In this section we   present several constructions of QTPCs based on Theorem \ref{CSS_TPC1_theorem} and  make   comparisons between QTPCs and other classes of QECCs. Firstly, we give two explicit examples to illustrate the construction of QTPCs.
\begin{table*}
\renewcommand{\arraystretch}{1.3}
\tcaption{Comparisons between QTPCs and CQCs derived from online code tables in \cite{Grassl:codetables}.}
\setlength{\tabcolsep}{5pt}
\label{Comparisons Between QTPCs and CQCs From Code Tables}
\centering
\footnotesize
\begin{tabular}[c]{|l|l|c|l|l|}
\hline
\multicolumn{1}{|c|}{\tabincell{c}{Component \\Codes $C_1$}} & \multicolumn{1}{c|}{\tabincell{c}{Inner Quantum\\Codes $Q_1$}}&$n_2$&\multicolumn{1}{c|}{QTPCs $\mathcal{Q}$} &\multicolumn{1}{c|}{CQCs $C_Q$ }\\
\hline
$[16,9,6]_4$&$[[16,14,2]]$&$8\leq n_2\leq 2^{14}+1$&$[[16n_2,16n_2-70,6]]$&
$[[16n_2,14n_2-56,6]]$\\
\hline
 $[18,9,8]_4$&$[[18,16,2]]$&$16\leq n_2\leq 2^{16}+1$&$[[18n_2,18n_2-126,8]]$&
$[[18n_2,16n_2-96,8]]$\\
\hline
 $[27,15,9]_4$&$[[27,20,3]]$&$17\leq n_2\leq 2^{20}+1$&$[[27n_2,27n_2-192,9]]$&
$[[27n_2,20n_2-80,9]]$\\
\hline
 $[28,15,10]_4$&$[[28,26,2]]$&$14\leq n_2\leq 2^{26}+1$&$[[28n_2,28n_2-234,10]]$&
$[[28n_2,26n_2-208,10]]$\\
\hline
 $[30,15,12]_4$&$[[30,23,3]]$&$28\leq n_2\leq 2^{23}+1$&$[[30n_2,30n_2-330,12]]$&
$[[30n_2,23n_2-138,12]]$\\
\hline
 $[44,22,14]_4$&$[[44,42,2]]$&$35\leq n_2\leq 2^{42}+1$&$[[44n_2,44n_2-572,14]]$&
$[[44n_2,42n_2-504,14]]$\\
\hline
 $[53,27,15]_4$&$[[53,45,3]]$&$47\leq n_2\leq 2^{45}+1$&$[[53n_2,53n_2-728,15]]$&
$[[53n_2,45n_2-360,15]]$\\
\hline
\end{tabular}
\end{table*}

\begin{examples}
Let $C_1=[5,3,3]_4$ be a Hermitian dual-containing code over $GF(4)$ in $\cite{calderbank1998quantum}$. Let $C_2=[n_2,n_2-2,3]_{16}$ be an arbitrary MDS code over $GF(16)$ in $\cite{macwilliams1981theory}$, where $3\leq n_2\leq 17$. Then the corresponding TPC $\mathcal{C} =C_2\otimes_H C_1=[5n_2,5n_2-4,3]_4$ is a Hermitian dual-containing code over $GF(4)$. Therefore, we can obtain a QTPC with parameters $\mathcal{Q}=[[5n_2,5n_2-8,3]]$  for $3\leq n_2\leq 17$. If $9\leq n_2\leq 17$,  then $\mathcal{Q}$ can attain the upper bound in  $\cite{Grassl:codetables}$.
\end{examples}

\begin{examples}
\label{QTPCs1}
Let $C_1=[n_1,1,n_1]$ be a binary repetition code. Let $C_2=[n_2,n_2-n_1+1,n_1]_{2^{\rho_1}}$ be  a narrow-sense RS code over $GF(2^{\rho_1})$, where $n_2=2^{\rho_1}-1$, $\rho_1=n_1-1$. Then the subfield subcode $\psi(C_2)$ of RS code $C_2$ is a narrow-sense BCH code $($see  $\cite{macwilliams1981theory})$.
If $n_1\leq2^{\lceil\frac{\rho_1}{2}\rceil}-1$, then $\psi(C_2)$ is dual contained according to $\cite{aly2007quantum}$. Therefore, there exists a QTPC with parameters $[[n_1n_2,n_1n_2-2(n_1-1)^2,  n_1]]$ by Theorem \ref{CSS_TPC1_theorem}.
 For example, we let $C_1=[9,1,9]$ and let $C_2=[255, 247,9]_{2^8}$ be  a narrow-sense
RS code over $GF(2^8)$, then there exists a QTPC with parameters $[[2295,2167, 9]]$. The rate is $0.94$. In order to construct a CQC with the
same length and minimum distance, we choose a binary QECC with parameters $[[15,9,3]]$ from $\cite{Grassl:codetables}$ as the inner QECC,
and choose a quantum MDS code over $GF(2^9)$ with parameters  $[[153,149,3]]$ as the outer QECC. Then there exists a CQC with
parameters $[[2295,1341,9]]$. The rate is $0.58$. It is shown that  the QTPC has a much higher code rate than the CQC.
\end{examples}

\subsection{\textbf{ QTPCs derived from dual-containing BCH codes}}
\noindent
In order to get larger numbers of comparable QTPCs and CQCs, we use BCH codes  as one of the component codes to construct QTPCs and use quantum BCH codes as the inner codes to construct CQCs. Here, we only consider the use of primitive, narrow-sense BCH codes and quantum  primitive, narrow-sense BCH codes, respectively.

Let $C_1$ be a binary primitive, narrow-sense BCH code with length $n_1=2^m-1$ and design distance $3\leq\delta_1\leq2^{\lceil m/2\rceil}-1$, and $m\geq 5$ is odd. According to  Ref. \cite{aly2007quantum}, we know that $C_1^{\bot}\subseteq C_1$ and $C_1$ has parameters $[n_1,n_1-\rho_1,\geq \delta_1]$, where the number of check symbols is $\rho_1=m\lceil1/2(\delta_1-1)\rceil$. Let $C_2=[n_2,n_2-\delta_1+1,\delta_1]_{2^{\rho_1}}$ be an MDS code over the extension field $GF(2^{\rho_1})$ and $\delta_1\leq n_2\leq 2^{\rho_1}+1$.  Then, there exists a QTPC with parameters $\mathcal{Q}=[[n_1n_2,n_1n_2-2\rho_1(\delta_1-1),\delta_1]]$ by Theorem \ref{CSS_TPC1_theorem}.

If $\delta_1\geq4$, then we let  $\delta_1=\eta_1\eta_2$ or $\delta_1-1=\eta_1\eta_2$, where $2\leq\eta_1\leq\eta_2$. Correspondingly, we construct a CQC which has the same length and similar minimum distance lower bound with the QTPC above. If $\eta_1=2$,   we let $D_1=[[n_1,K_1,2]]$ be an optimal stabilizer code with dimension $K_1=n_1-3$ according to Ref. \cite{rains1999quantum}. If $\eta_1>2$, we let $D_1=[[n_1,K_1,\geq \eta_1]]$ be a quantum BCH code which has  dimension  $K_1=n_1-2m\lceil1/2(\eta_1-1)\rceil$ and    design distance $\eta_1$  by Theorem 21 in \cite{aly2007quantum}. We let $D_2=[[n_2,n_2-2\eta_2+2,\eta_2]]_{2^{K_1}}$ be a quantum MDS  code over the extension field $GF(2^{K_1})$ and $\eta_2\leq n_2\leq 2^{K_1}+1$ according to Ref. \cite{rotteler2004quantum}. Then we have $\delta_1\leq n_2\leq\min\{2^{\rho_1}+1,2^{K_1}+1\}=2^{\rho_1}+1$. Hence, there exist a CQC with parameters $C_Q=[[n_1n_2,K_1(n_2-2\eta_2+2),\eta_1\eta_2]]$.
 It is easy to verify that $K_1(n_2-2\eta_2+2)< n_1n_2-2\rho_1(\delta_1-1)$ if $n_2\geq \lceil(1-\frac{2}{m})n_1\rceil$. Thus, QTPCs have larger dimension than     CQCs if $\lceil(1-\frac{2}{m})n_1\rceil\leq n_2\leq 2^{\rho_1}+1$. In order to guarantee that QTPCs  have larger dimension than the CQCs, we make a loose estimation of lower bound of $n_2$ here.

 If $m\geq 4$ is even and $C_1$ is a  quaternary  primitive, narrow-sense BCH code with length  $n_1=4^{m/2}-1$, similar results can be obtained by using the Hermitian construction.  In Table \ref{Comparisons Between QTPCs and CQCs I}, we compute and compare the parameters of QTPCs constructed based on  dual-containing BCH codes and those of  CQCs. It is easy to see that all the QTPCs in Table   \ref{Comparisons Between QTPCs and CQCs I} have parameters better than the CQCs, and when code length $n_2$ becomes larger, QTPCs have much larger dimension than CQCs with the same length and minimum distance. Notice that the lower bound of $n_2$  in Table \ref{Comparisons Between QTPCs and CQCs I} is much smaller than $\lceil(1-\frac{2}{m})n_1\rceil$.

\begin{table*}
\renewcommand{\arraystretch}{1.6}
\setlength{\tabcolsep}{4pt}
\tcaption{Comparisons between QTPCs and quantum BCH codes. The QPTCs have component codes $C_1$ and $C_2$, where   $C_1$ is chosen from online code tables in \cite{Grassl:codetables} and the MAGMA database, and $C_2$ is an MDS code over the extension field from \cite{macwilliams1981theory}. The quantum BCH codes are derived from Ref. \cite{aly2007quantum,ling2010generalization}.}
\label{Comparisons Between QTPCs and CQCs II}
\centering
\tiny
\begin{tabular}[c]{|l|l|l|l|}
\hline
\multicolumn{1}{|c|}{\tabincell{c}{Component \\Codes $C_1$}}&\multicolumn{1}{c|}{\tabincell{c}{Component \\Codes $C_2$}}&\multicolumn{1}{c|}{QTPCs $\mathcal{Q}$ \&  Rates $R$ }&\multicolumn{1}{c|}{Quantum BCH codes $Q_B$ \&  Rates $ R$ }\\
\hline
$[6,1,6]_4$&$[341,337,5]_{4^5}$& $ [[2^{11}-2,2^{11}-42,5]]$, $R\approx0.97949$ &$ [[2^{11}-1,2^{11}-45,\geq5]]$, $R\approx0.97851$ \\
$[6,1,6]_4$&$[1023,1018,6]_{4^5}$& $ [[6138,6088,6]]$, $R\approx0.99185$ &$ [[6223,6139,\geq5]]$, $R\approx0.9865$ \\
\hline
$[11,6,5]_4$&$[744,740,5]_{4^5}$& $ [[2^{13}-8,2^{13}-48,5]]$, $R\approx0.99511$ &$ [[2^{13}-1,2^{13}-53,\geq5]]$, $R\approx0.99365$ \\
$[11,6,5]_4$&$[497,493,5]_{4^5}$& $ [[\frac{2^{14}-1}{3}+6,\frac{2^{14}-1}{3}-34,5]]$, $R\approx0.99268$ &$ [[\frac{2^{14}-1}{3},\frac{2^{14}-1}{3}-42,\geq5]]$, $R\approx0.99231$ \\
$[11,6,5]_4$&$[425,493,5]_{4^5}$& $ [[\frac{2^{15}-1}{7}-6,\frac{2^{15}-1}{7}-46,5]]$, $R\approx0.99017$ &$ [[\frac{2^{15}-1}{7},\frac{2^{15}-1}{7}-60,\geq5]]$, $R\approx0.98718$ \\
$[11,6,5]_4$&$[397,393,5]_{4^5}$& $ [[\frac{2^{16}-1}{15}-2,\frac{2^{16}-1}{15}-42,5]]$, $R\approx0.99084$ &$ [[\frac{2^{16}-1}{15},\frac{2^{16}-1}{15}-48,\geq5]]$, $R\approx0.98901$ \\
$[11,6,5]_4$&$[350,346,5]_{4^5}$& $ [[\frac{2^{16}-1}{17}-5,\frac{2^{16}-1}{17}-45,5]]$, $R\approx0.98961$ &$ [[\frac{2^{16}-1}{17},\frac{2^{16}-1}{17}-48,\geq5]]$, $R\approx0.98755$ \\
\hline
$[12,6,6]_4$&$[2731,2726,6]_{4^6}$& $ [[2^{15}+5,2^{15}-55,6]]$, $R\approx0.99817$ &$ [[2^{15}-1,2^{15}-61,\geq5]]$, $R\approx0.99817$ \\
\hline
$[14,8,5]_4$&$[75,71,5]_{4^6}$& $ [[\frac{2^{15}-1}{31}-7,\frac{2^{15}-1}{31}-55,5]]$, $R\approx0.95429$ &$ [[\frac{2^{15}-1}{31},\frac{2^{15}-1}{31}-60,\geq5]]$, $R\approx0.94324$ \\
$[14,8,5]_4$&$[334,330,5]_{4^6}$& $ [[\frac{2^{15}-1}{7}-5,\frac{2^{15}-1}{7}-53,5]]$, $R\approx0.98973$ &$ [[\frac{2^{15}-1}{7},\frac{2^{15}-1}{7}-60,\geq5]]$, $R\approx0.98718$ \\
$[14,8,5]_4$&$[2340,2336,5]_{4^6}$& $ [[2^{15}-8,2^{15}-56,5]]$, $R\approx0.99853$ &$ [[2^{15}-1,2^{15}-61,\geq5]]$, $R\approx0.99817$ \\
\hline

$[17,9,7]_4$&$[663,659,7]_{4^8}$& $ [[\frac{2^{20}-1}{93}-4,\frac{2^{20}-1}{93}-100,7]]$, $R\approx0.99148$ &$ [[\frac{2^{20}-1}{93},\frac{2^{20}-1}{93}-100,\geq7]]$, $R\approx0.99113$ \\

$[17,9,7]_4$&$[823,816,7]_{4^8}$& $ [[\frac{2^{20}-1}{75}+10,\frac{2^{20}-1}{75}-86,7]]$, $R\approx0.99314$ &$ [[\frac{2^{20}-1}{75},\frac{2^{20}-1}{75}-100,\geq7]]$, $R\approx0.99285$ \\

$[17,9,7]_4$&$[1869,1863,7]_{4^8}$& $ [[\frac{2^{20}-1}{33}-2,\frac{2^{20}-1}{33}-98,7]]$, $R\approx0.99692$ &$ [[\frac{2^{20}-1}{33},\frac{2^{20}-1}{33}-100,\geq7]]$, $R\approx0.99685$ \\

$[17,9,7]_4$&$[7710,7704,7]_{4^8}$& $ [[2^{20}+1,2^{20}-95,7]]$, $R\approx0.99991$ &$ [[2^{20}-1,2^{20}-101,\geq7]]$, $R\approx0.9999$ \\

$[17,9,7]_4$&$[971,965,7]_{4^8}$& $ [[\frac{2^{21}-1}{127}-6,\frac{2^{21}-1}{127}-102,7]]$, $R\approx0.99418$ &$ [[\frac{2^{21}-1}{127},\frac{2^{21}-1}{127}-105,\geq7]]$, $R\approx0.99364$ \\

$[17,9,7]_4$&$[2518,2512,7]_{4^8}$& $ [[\frac{2^{21}-1}{49}+7,\frac{2^{21}-1}{49}-89,7]]$, $R\approx0.99776$ &$ [[\frac{2^{21}-1}{49},\frac{2^{21}-1}{49}-105,\geq7]]$, $R\approx0.99755$ \\

$[17,9,7]_4$&$[17623,17617,7]_{4^8}$& $ [[\frac{2^{21}-1}{7}-2,\frac{2^{21}-1}{7}-98,7]]$, $R\approx0.99968$ &$ [[\frac{2^{21}-1}{7},\frac{2^{21}-1}{7}-105,\geq7]]$, $R\approx0.99965$ \\
\hline
$[27,15,9]_4$&$[9519,9511,9]_{4^{12}}$ & $ [[\frac{2^{29}-1}{2089}+14,\frac{2^{29}-1}{2089}-178,9]]$, $R\approx0.99925$ &$ [[\frac{2^{29}-1}{2089},\frac{2^{29}-1}{2089}-203,\geq9]]$, $R\approx0.99921$ \\
$[27,15,9]_4$&$[18027,18019,9]_{4^{12}}$& $ [[\frac{2^{29}-1}{1103}-8,\frac{2^{29}-1}{1103}-200,9]]$, $R\approx0.99961$ &$ [[\frac{2^{29}-1}{1103},\frac{2^{29}-1}{1103}-203,\geq9]]$, $R\approx0.99958$ \\
$[27,15,9]_4$&$[85340,85332,9]_{4^{12}}$& $ [[\frac{2^{29}-1}{233}+13,\frac{2^{29}-1}{233}-179,9]]$, $R\approx0.99992$ &$ [[\frac{2^{29}-1}{233},\frac{2^{29}-1}{233}-203,\geq9]]$, $R\approx0.99991$ \\
\hline
\end{tabular}
\end{table*}

\subsection{\textbf{QTPCs with component codes derived from online code tables and MAGMA}}
\noindent
The online code tables in \cite{Grassl:codetables} provide bounds on the parameters of classical linear codes and additive quantum codes. Almost all of those codes are Best Known Linear Codes (BKLC) or Best Known Quantum Codes (BKQC). Therefore, we can construct QTPCs with component codes $C_1$  chosen from the code tables in \cite{Grassl:codetables}.    From Theorem \ref{additive_quantum_codes}, we know that each  additive code $\mathbf{Q}=[[n,k]]$ in \cite{Grassl:codetables} corresponds to a classical  additive code $\mathbf{D}=[n,(n-k)/2]_4$ that is self-orthogonal with respect to the  trace-Hermitian inner product over $GF(4)$, i.e., $\mathbf{D}\subseteq \mathbf{D}^{\bot_{th}}$. In particular, if code $\mathbf{D}$ is linear over $GF(4)$,  then  $\mathbf{D}^{\bot_{th}}$ is equal to the Hermitian dual code  of code $\mathbf{D}$, i.e., $\mathbf{D}^{\bot_{th}}=\mathbf{D}^{\bot_{h}}$. Let $\mathbf{C}=\mathbf{D}^{\bot_{h}}$, then $\mathbf{C}$ is a Hermitian dual-containing code, i.e., $\mathbf{C}^{\bot_{h}}\subseteq \mathbf{C}$. Furthermore, if $\mathbf{Q}$ is pure and has minimum distance $d$, then $\mathbf{C}$ has minimum distance $d$.

In Table \ref{Comparisons Between QTPCs and CQCs From Code Tables}-\ref{Comparisons Between QTPCs and QECCs III}, we construct many QTPCs based on the code tables in \cite{Grassl:codetables} and MAGMA \cite{cannon2008handbook} (Version 2.21-8, online).  We first choose a BKQC $\mathbf{Q}$ with length less than or equal to $50$ from \cite{Grassl:codetables} or the MAGMA database, then we can get the corresponding classical additive code  $\mathbf{D}$ over $GF(4)$. By using MAGMA, we can determine that if quantum code $\mathbf{Q}$ is pure and if code $\mathbf{D}$ is linear over $GF(4)$. If  quantum code $\mathbf{Q}$ is pure and code $\mathbf{D}$ is linear over $GF(4)$, then we get a Hermitian dual-containing code $\mathbf{C}=\mathbf{D}^{\bot_h}$ which has the same minimum distance with $\mathbf{Q}$.  All the component codes $C_1$ in Table \ref{Comparisons Between QTPCs and CQCs From Code Tables}-\ref{Comparisons Between QTPCs and QECCs III}  except   $C_1=[6,1,6]_4$, $C_1=[11,6,5]$ and $C_1=[12,6,6]$ in Table  \ref{Comparisons Between QTPCs and CQCs II} and Table \ref{Comparisons Between QTPCs and QECCs III} are trace-Hermitian dual-containing additive codes that correspond  to   BKQCs $\mathbf{Q}$ in \cite{Grassl:codetables} and the MAGMA database. By  using MAGMA, we know that they happen to be linear   and Hermitian dual-containing codes over $GF(4)$, and have the same minimum distance  with the corresponding BKQCs $\mathbf{Q}$. The codes $C_1=[11,6,5]_4$ and $C_1=[12,6,6]$   in Table \ref{Comparisons Between QTPCs and CQCs II} and Table \ref{Comparisons Between QTPCs and QECCs III} are    BKLCs in the database of MAGMA and it is easy to verify that they are also Hermitian  dual-containing codes. All the component codes $C_2$ used for the construction of QTPCs in Table \ref{Comparisons Between QTPCs and CQCs From Code Tables}-\ref{Comparisons Between QTPCs and QECCs III}  are MDS codes over the extension field by
Ref. \cite{macwilliams1981theory}. For the component codes $C_2=[341,337,5]_{4^5}$ and $C_2=[1023,1018,6]_{4^5}$ in Table \ref{Comparisons Between QTPCs and CQCs II}, it is easy to verify that
$\psi(C_2)^{\bot_h}\subseteq \psi(C_2)$ according to \cite{aly2007quantum}. In Table \ref{Comparisons Between QTPCs and CQCs From Code Tables}, we use quantum MDS codes from \cite{rotteler2004quantum} as the outer QECCs  of CQCs. Therefore, we can obtain a lot of QTPCs from Theorem \ref{CSS_TPC1_theorem}.

All the QTPCs in Table \ref{Comparisons Between QTPCs and CQCs From Code Tables} have better parameters than CQCs with the same length and minimum distance, and when the component code length $n_2$ becomes larger, QTPCs have much higher dimension than CQCs with the same length and minimum distance.
All the  QTPCs in  Table \ref{Comparisons Between QTPCs and CQCs II} except  the second one  have exactly similar code length and minimum distance with quantum BCH codes in \cite{aly2007quantum,ling2010generalization}, but have higher code rates. The second QTPC in Table \ref{Comparisons Between QTPCs and CQCs II} has a larger difference of code length with the quantum BCH code, but have a higher code rate and larger minimum distance, then we still say that it is better than the quantum BCH code.  All the QTPCs in Table \ref{Comparisons Between QTPCs and QECCs III} have very similar code length  and the same minimum distance with the comparable QECCs in \cite{li2008binary}, but have larger dimension.  Therefore, all the QTPCs in Table \ref{Comparisons Between QTPCs and CQCs II} and \ref{Comparisons Between QTPCs and QECCs III} have   parameters   better than the ones available in the previous literature.

\begin{table*}
\renewcommand{\arraystretch}{1.3}
\setlength{\tabcolsep}{2pt}
\tcaption{Comparisons between QTPCs and QECCs in Ref. \cite{li2008binary}. The QPTCs have component codes $C_1$ and $C_2$, where  $C_1$ is chosen from online code tables in \cite{Grassl:codetables} and the MAGMA Database, and   $C_2$ is an MDS Code  over the extension field from \cite{macwilliams1981theory}. According to \cite{li2008binary}, $N_1(r)=2(2^r-1)/3$, $N_2(r)=6(2^r-1))/7$, and $N_3(r)=4(2^r-1)/5$.}
\label{Comparisons Between QTPCs and QECCs III}
\centering
\tiny
\begin{tabular}[c]{|c|l|c|l|l|}
\hline
$r$&\tabincell{c}{Component \\Codes $C_1$}&\tabincell{c}{Component \\Codes $C_2$}&\multicolumn{1}{c|}{QTPCs $\mathcal{Q}$ }&  \multicolumn{1}{c|}{QECCs $Q$ in Ref. \cite{li2008binary} }   \\
\hline
10& $[11,6,5]_4$& $[155,151,5]_{4^5}$& $ [[ 2^r+N_1(r)-1,2^r+N_1(r)-41,5]]$ &$ [[2^r+N_1(r),2^r+N_1(r)-42,5]]$ \\
\hline
10& $[11,6,5]_4$& $[124,120,5]_{4^5}$& $ [[ 2N_1(r),2N_1(r)-40,5]]$ &$ [[2N_1(r),2N_1(r)-42,5]]$ \\
\hline
12& $[11,6,5]_4$& \tabincell{c}{$[n_2,n_2-4,5]_{4^5}$,\\$n_2=\lceil\frac{2^t+N_i(r)}{11}\rceil$, or\\ $n_2=\lfloor\frac{2^t+N_i(r)}{11}\rfloor$,\\$1\leq i\leq 3$, $5\leq t\leq r$}& $ [[ 11n_2,11n_2-40,5]]$ &$ [[2^t+N_i(r),2^t+N_i(r)-38-t,5]]$ \\
\hline
12& $[11,6,5]_4$& \tabincell{c}{$[n_2,n_2-4,5]_{4^5}$,\\$n_2=\lceil\frac{N_i(r)+N_j(r)}{11}\rceil$, or\\$n_2=\lfloor\frac{N_i(r)+N_j(r)}{11}\rfloor$, \\$1\leq i,j\leq 3$}& $ [[ 11n_2,11n_2-40,5]]$ &$ [[N_i(r)+N_j(r),N_i(r)+N_j(r)-50,5]]$ \\
\hline
12& $[11,6,5]_4$& \tabincell{c}{$[n_2,n_2-4,5]_{4^5}$,\\$n_2=\lceil\frac{24+N_i(r)}{11}\rceil$, or\\$n_2=\lfloor\frac{24+N_i(r)}{11}\rfloor$, \\$1\leq i\leq 3$}& $ [[ 11n_2,11n_2-40,5]]$ &$ [[24+N_i(r),24+N_i(r)-43,5]]$ \\
\hline
14& $[11,6,5]_4$&  $[993,989,5]_{4^5}$ & $ [[N_1(r)+1,N_1(r)-39,5]]$ &$ [[N_1(r),N_1(r)-42,5]]$ \\
\hline
14& $[11,6,5]_4$& \tabincell{c}{$[n_2,n_2-4,5]_{4^5}$,\\$n_2=\lceil\frac{2^t+N_1(r)}{11}\rceil$, or\\$n_2=\lfloor\frac{2^t+N_1(r)}{11}\rfloor$, \\ $5\leq t\leq 8$}& $ [[ 11n_2,11n_2-40,5]]$ &$ [[2^t+N_1(r),2^t+N_1(r)-44-t,5]]$ \\
\hline
14& $[11,6,5]_4$& $[995,991,5]_{4^5}$ & $ [[23+N_1(r),N_1(r)-17,5]]$ &$ [[24+N_1(r),N_1(r)-25,5]]$ \\
\hline
15& $[14,8,5]_4$&  $[2008,2004,5]_{4^6}$ & $ [[26+N_2(r),N_2(r)-22,5]]$ &$ [[24+N_2(r),N_2(r)-28,5]]$ \\
\hline
15& $[14,8,5]_4$& \tabincell{c}{$[n_2,n_2-4,5]_{4^5}$,\\$n_2=\lceil\frac{2^t+N_2(r)}{14}\rceil$, or\\$n_2=\lfloor\frac{2^t+N_2(r)}{14}\rfloor$, \\ $5\leq t\leq r-1$}& $ [[ 14n_2,14n_2-48,5]]$ &$ [[2^t+N_2(r),2^t+N_2(r)-47-t,5]]$ \\
\hline
15& $[14,8,5]_4$&$[4012,4008,5]_{4^5}$& $ [[2N_2(r)-4,2N_2(r)-52,5]]$ &$ [[2N_2(r),2N_2(r)-62,5]]$ \\
\hline
15& $[12,6,6]_4$& \tabincell{c}{$[n_2,n_2-5,6]_{4^6}$,\\$n_2=\lceil\frac{2^t+N_2(r)}{12}\rceil$, or\\$n_2=\lfloor\frac{2^t+N_2(r)}{12}\rfloor$,\\ $12\leq t\leq r$}& $ [[ 12n_2,12n_2-60,6]]$ &$ [[2^t+N_2(r),2^t+N_2(r)-49-t,6]]$ \\
\hline
16& $[12,6,6]_4$& \tabincell{c}{$[n_2,n_2-5,6]_{4^6}$,\\$n_2=\lceil\frac{2^t+N_1(r)}{12}\rceil$, or\\$n_2=\lfloor\frac{2^t+N_1(r)}{12}\rfloor$,\\ $9\leq t\leq 12$}& $ [[ 12n_2,12n_2-60,6]]$ &$ [[2^t+N_1(r),2^t+N_1(r)-52-t,6]]$ \\
\hline
\end{tabular}
\end{table*}

\subsection{\textbf{QTPCs with self-dual component codes}}
\noindent
From Lemma \ref{TPC_Same_Fields}, it is easy to see that if there is a linear code satisfying $C_1^{\bot} \subseteq C_1$ (or $C_1^{\bot_h} \subseteq C_1$), we can concatenate it with an arbitrary linear code $C_2$ over the same field to get a dual-containing TPC.  As a special case, if we  use self-dual codes to construct TPCs,  the following QTPCs can be constructed.

\begin{corollary}
\label{sel_dual_same_field}
Let $C$ be an arbitrary self-dual code over $GF(q)$$(q=2$ or $4)$   with parameters $[n,n/2 ,d]$. Then there exists a binary QTPC with parameters $[[n^2,n^2/2 ,d]]$, and  $d\geq H_q^{-1}(1/2)n$    when $n\rightarrow \infty$.
\end{corollary}
\begin{proof}
 From   Lemma \ref{TPC_Same_Fields} and Ref. \cite{grassl2005quantum}, we know that there exists a binary QTPC with parameters $[[n^2,n^2/2 ,d]]$.
According to Ref. \cite[Ch. 19]{macwilliams1981theory},  there exist long $q$-ary self-dual codes which achieve  the Gilbert-Varshamov bound, i.e.,  $d\geq H_q^{-1}(1/2)n$    when $n\rightarrow \infty$, where $H_q^{-1}(x)$ is the inverse of the entropy function. \qed
\end{proof}

If we combine binary or quaternary self-dual codes with an arbitrary MDS code over the extension field to construct TPCs, we can obtain QTPCs with   better parameters than those in Corollary \ref{sel_dual_same_field}.  Let $C_1$ be an arbitrary self-dual code over $GF(q)$    with parameters $[n,n/2 ,d]$. Let $C_2$ be an MDS code over the extension field with parameters $[n,n-d+1,d]$.   Then we have the following result.

\begin{corollary}
\label{self_dual_QTPCs2}
There exists a  QTPC with parameters $[[n^2,n^2-nd+n,d]]$, and  $d\geq H_q^{-1}(1/2)n$    when $n\rightarrow \infty$.
\end{corollary}

The proof of Corollary \ref{self_dual_QTPCs2} is similar to that of Corollary \ref{sel_dual_same_field}. The rate of the QTPC  in Corollary \ref{self_dual_QTPCs2} is $r=1-\frac{d-1}{n}$ which is   higher than $1/2$ in Corollary \ref{sel_dual_same_field}.

\section{QTPCs With Burst-Error-Correction Abilities}
\label{burst QTPCs}
\noindent
In this section we use cyclic component codes and classical MDS  codes to construct QTPCs with multiple-burst-error-correction abilities based on Theorem \ref{CSS_TPC2}. The restriction on $C_2$ in Theorem \ref{CSS_TPC2} is much easier to be satisfied than that in Theorem \ref{CSS_TPC1_theorem}, but one condition  is that the product matrix $H_{c_1}H_{c_1}^T$ needs to to be a full rank.  If we let $C_1$ be a binary cyclic code with defining set $Z_1$ and $\gcd(n_1,2)=1$, it is easy to verify that this condition on $C_1$ is equivalent to $Z_1=Z_1^{-1}$, where $Z_1^{-1}=\{-z\pmod{n_1}|z\in Z_1\}$. Such cyclic codes are called \emph{reversible} codes in \cite{macwilliams1981theory,massey1964reversible}. Denote by
$f_r(x)\equiv x^{\deg f(x)}f(x^{-1})$
the \emph{reciprocal} polynomial of $f(x)$. A monic polynomial $f(x)$ will be called \emph{self-reciprocal} if and only if $f(x)=f_r(x)$.

\begin{lemma}[{\cite[Theorem 1]{massey1964reversible}}]
The cyclic code generated by the monic polynomial $g(x)$ is
reversible if and only if  $g(x)$ is self-reciprocal.
\end{lemma}

We denote by $\mathcal{BCH}(n,b;\delta)$ a binary  BCH code with design distance $\delta$ and   defining
 set $Z=\{\mathfrak{C}_b,\mathfrak{C}_{b+1},\ldots,\mathfrak{C}_{b+\delta-2}\}$, where $\mathfrak{C}_b=\{b2^s\pmod{n}| s\in \mathbb{Z}, s\geq 0\}$ denotes
 the binary cyclotomic
coset of $b$ mod $n$, $0\leq b\leq n-1$.  For every design distance $2\leq \delta\leq n$, there always exists at least one $  b\in\{0,1,\ldots, n-1\}$
such that $Z=Z^{-1}$ from \cite{macwilliams1981theory}. Therefore, we can always find out a reversible BCH code with design distance $2\leq \delta\leq n$
and use it as the binary component code of the TPC.

\begin{theorem}
\label{Reversible_BCH}
Let $n_1$ be an odd integer. Let $C_1=[n_1,1,n_1]$ be a repetition code. Let $C_2=[n_2,n_2-n_1+1,n_1]_{2^{\rho_1}}$ be an MDS code over $GF(2^{\rho_1})$, where $n_2\leq2^{\rho_1}$, $\rho_1=n_1-1$. If $n_1\leq\lfloor n_2/2\rfloor+1$, then there exists a QTPC with parameters $[[n_1n_2,n_1n_2-2(n_1-1)^2, n_1]]$. This code has a quantum analog of multiple $(\lceil\frac{n_1}{2}\rceil-1)$-burst-error-correction abilities, provided these bursts fall in distinct subblocks.
\end{theorem}
\begin{proof}
If $n_1$ is odd, it is easy to see that $C_1=[n_1,1,n_1]$ is a reversible   code.  If the design distance $n_1\leq\lfloor n_2/2\rfloor+1$ and $n_2\leq2^{\rho_1}$, then  there exists a dual-containing MDS
code $C_2=[n_2,n_2-n_1+1,n_1]_{2^{\rho_1}}$ over $GF(2^{\rho_1})$ (see \cite{li2008quantum}).
Therefore, there exists a QTPC with parameters $[[n_1n_2,n_1n_2-2(n_1-1)^2, n_1]]$ by Theorem \ref{CSS_TPC2}. Combining Lemma \ref{vatan_Constructions} and Lemma \ref{wolftheorem}, we know that this code has a quantum analog of
multiple $(\lceil\frac{n_1}{2}\rceil-1)$-burst-error-correction abilities, provided these bursts fall in distinct subblocks.  \qed
\end{proof}

Fire codes  are a class of cyclic codes  used for correcting burst errors \cite{lin2004error}.  The definition of Fire codes is given as follows. Let $b(x)$ be an irreducible polynomial of degree $w$ over $GF(2)$. Let $\rho$ be the period of $b(x)$ where $\rho$ is the smallest integer such that $b(x)$ divides $x^\rho+1$. Let $l$
 be a positive integer such that $l\leq w$, and $2l-1$ is not divisible by $\rho$. Then an $l$-burst-error-correction Fire code $F$  is defined by the   generator polynomial  \[g(x)=(x^{2l-1}+1)b(x).\]
The length $n$ of this Fire code is the least common multiple (LCM) of $2l-1$ and the period of $\rho$ of $b(x)$, i.e.,
\[n=\text{LCM}(2l-1,\rho).\]
The number of   check symbols of this  code is $w+2l-1$. Note that the two factors $x^{2l-1}+1$ and $b(x)$ are relatively prime.

It is easy to see that the $(x^{2l-1}+1)$-factor in $g(x)$ is a self-reciprocal polynomial. If we choose $b(x)$ as a self-reciprocal irreducible polynomial over $GF(2)$, then $g(x)$ is   a self-reciprocal  polynomial over $GF(2)$. In \cite{yucas2004self},  the number of self-reciprocal irreducible polynomials of degree $w=2t$ over $GF(2)$ is given by
\begin{equation}
N_2(w)=\frac{1}{2t}\sum_{d|t,\ d\text{ odd}}\mu(d)2^{t/d}.
\end{equation}
It is easy to verify that $N_2(w)>0$. Therefore, we can always choose a self-reciprocal polynomial with an even degree to construct a reversible Fire code. Then we can construct QTPCs based on reversible Fire codes with
multiple-burst-error-correction abilities.

\begin{theorem}
 \label{Fire_RS_QTPC}
 Let  $C_1=[n_1,k_1]$ be a reversible $l$-burst-error-correction  Fire code, and let $C_2=[n_2,k_2,n_2-k_2+1]_{2^{\rho_1}}$ be  a RS code  over the extension field $GF(2^{\rho_1})$, and the numbers of check symbols are  $\rho_1=n_1-k_1$ and  $\rho_2=n_2-k_2$, respectively. If
 $k_2 \geq \lceil\frac{n_2}{2}\rceil$, then there exists a QTPC with parameters $\mathcal{Q}=[[n_1n_2,n_1n_2-2\rho_1\rho_2]]$ which has a quantum analog of  $\lfloor\frac{\rho_2+1}{2}\rfloor$ numbers of $l$-burst-error-correction abilities, provided these bursts fall in distinct subblocks.
 \end{theorem}
 \begin{proof}
Let $H_{c_1}$ be the parity check matrix of $C_1$, then the product matrix $H_{c_1}H_{c_1}^T$ is of full rank if
$C_1$ is a reversible Fire code. If $k_2 \geq \lceil\frac{n_2}{2}\rceil$, then there exists a dual-containing RS code $C_2^\perp\subseteq C_2$ according to \cite{grassl1999quantum}.
 Let $\mathcal{C}=C_2\otimes_H C_1$ be the tensor product code of $C_1$ and $C_2$, then $\mathcal{C}$ is an
$[n_1n_2,n_1n_2-\rho_1\rho_2]$ binary code which corrects $\lfloor\frac{n_2-k_2+1}{2}\rfloor=\lfloor\frac{\rho_2+1}{2}\rfloor$ or fewer bursts
of errors, each burst is less than or equal to $l$, provided these
bursts fall in distinct subblocks.
Combining Lemma \ref{vatan_Constructions},  Lemma \ref{wolftheorem} and Theorem \ref{CSS_TPC2}, we know that there exists a QTPC with parameters $\mathcal{Q}=[[n_1n_2,n_1n_2-2\rho_1\rho_2]]$  which has a quantum analog of  $\lfloor\frac{\rho_2+1}{2}\rfloor$ numbers of $l$-burst-error-correction abilities, provided these bursts fall in distinct subblocks. \qed
 \end{proof}

\begin{examples}
Consider the self-reciprocal polynomial $b(x)=1+x^2+x^3+x^4$ over $GF(2)$. Its period is $\rho=5$. Let $l=4$. A reversible Fire code is defined by the generator polynomial  $g(x)=(x^{7}+1)b(x)$. Its length is $n=\text{LCM}(7,5)=35$. Then a  $[35,24]$ reversible Fire code $F$ with $4$-burst-error-correction abilities  can be obtained. Choose $C=[23,24-t,t]$ as a narrow-sense RS code over $GF(2^{11})$. It is easy to see that if $2\leq t\leq 12$, then $C^\bot \subseteq C$  from $\cite{grassl1999quantum}$. Let $\mathcal{C}=C\otimes_H F$ be the tensor product code of $C$ and $F$, then  $\mathcal{C}$ is an $[805,816-11t]$  binary
code which corrects $\lfloor\frac{t}{2}\rfloor$ or fewer bursts of errors, each
burst is less than or equal to $4$, provided these bursts
fall in distinct subblocks. Then there exists a QTPC  $\mathcal{Q}=[[805,827-22t]]$ with a quantum analog of $\lfloor\frac{t}{2}\rfloor$  numbers of $4$-burst-error-correction abilities, provided these bursts fall in distinct subblocks.
\end{examples}

\section{Decoding Of QTPCs}
\label{QTPCs Decoding}
\noindent
The decoding  procedure for the classical TPCs can be done by performing  an outer decoding firstly, followed by an inner decoding    \cite{wolf1965codes,alhussien2010iteratively}.
The decoding of QTPCs can be done similarly. 
\begin{itemize}
\item[i)] Outer Decoding:
Through performing  outer measurement on the ancilla qubits, the syndrome
is calculated as a $q$-ary vector  and is mapped
to a vector  with subblocks over $GF(q^{\rho_1})$. Then the outer recovery and decoding are performed.
\item[ii)] Inner Decoding: If the outer decoding is successfully accomplished,  the erroneous
subblocks could be determined. Then measurement is performed  on the ancilla qubits correlated with the erroneous subblocks. Finally, the inner decoding
is performed  only on the erroneous subblocks.
\end{itemize}

\section{Conclusion and Discussion}
\label{Conclusions}
\noindent
A general construction    of QTPCs was proposed in this paper. Since the parity check matrix of a classical TPC has a tensor product structure, the
construction of the corresponding QTPC  is less constrained compared to   other classes of QECCs, leading to
many choices of selecting component  codes in designing various QTPCs.
Compared with  CQCs, the component code selections of QTPCs are much
more flexible than those of CQCs. Several families of QTPCs have been constructed with parameters better than other classes of QECCs.
  It is worth noting that all QTPCs constructed are pure and have exact parameters. In particular,
QTPCs have  quantum multiple-burst-error-correction abilities as
their classical counterparts,  provided these bursts fall in distinct
subblocks.
Finally, whether QTPCs could be used in   quantum storage systems just like what happened of classical TPCs is worthy of further investigation in future.

%
\section*{Acknowledgments}
\noindent
The authors are grateful to the Editor and the anonymous referees for the constructive
comments and valuable suggestions that help to improve
the manuscript. The work of Jihao Fan was supported by the China Scholarship Council (Grant No. 201406090079), the National Natural Science
Foundation of China (Grant No. 61403188), and the Natural Science Foundation
of Jiangsu Province (Grant No. BK20140823). The work of Yonghui Li was supported by an ARC under Grant DP150104019.  MH was supported by an ARC Future Fellowship under Grant FT140100574.

\newpage
\nonumsection{References}
\noindent

\bibliographystyle{IEEEtranS}
\bibliography{IEEEabrv,QIC160802}

\end{document}